\documentclass[11pt]{article}%
\usepackage{amsmath}
\usepackage{amsfonts}
\usepackage{times}
\usepackage{fullpage}
\usepackage{amssymb}
\usepackage{graphicx}%
\providecommand{\U}[1]{\protect\rule{.1in}{.1in}}
\newtheorem{theorem}{Theorem}

\newtheorem{conjecture}[theorem]{Conjecture}
\newtheorem{corollary}[theorem]{Corollary}

\newtheorem{lemma}[theorem]{Lemma}

\newtheorem{problem}[theorem]{Problem}
\newtheorem{proposition}[theorem]{Proposition}

\newenvironment{proof}[1][Proof]{\noindent\textbf{#1.} }{\ \rule{0.5em}{0.5em}}
\begin{document}

\title{Weak Parity}
\author{Scott Aaronson\thanks{MIT. \ aaronson@csail.mit.edu. \ Supported by the
National Science Foundation under Grant No.\ 0844626, a TIBCO Chair, and an
Alan T.\ Waterman Award.}
\and Andris Ambainis\thanks{University of Latvia. \ andris.ambainis@lu.lv.
\ Supported by the European Commission under the project QALGO (Grant
No.\ 600700) and the ERC Advanced Grant MQC (Grant No.\ 320731).}
\and Kaspars Balodis\thanks{University of Latvia. \ kbalodis@gmail.com. \ Supported
by the European Social Fund within the project \textquotedblleft Support for
Doctoral Studies at University of Latvia.\textquotedblright}
\and Mohammad Bavarian\thanks{MIT. \ bavarian@mit.edu. \ Partially supported by NSF
STC Award 0939370.}}
\date{}
\maketitle

\begin{abstract}
We study the query complexity of \textsc{Weak Parity}: the problem of
computing the parity of an $n$-bit input string, where one only has to succeed
on a $1/2+\varepsilon$ fraction of input strings, but must do so with high
probability on those inputs where one does succeed. \ It is well-known that
$n$ randomized queries and $n/2$ quantum queries are needed to compute parity
on \textit{all} inputs. \ But surprisingly, we give a randomized algorithm for
\textsc{Weak Parity}\ that makes only $O(n/\log^{0.246}(1/\varepsilon
))$\ queries, as well as a quantum algorithm that makes $O(n/\sqrt
{\log(1/\varepsilon)})$ queries.\ \ We also prove a lower bound of
$\Omega\left(  n/\log\left(  1/\varepsilon\right)  \right)  $ in both cases;
and using extremal combinatorics, prove lower bounds of $\Omega(\log n)$\ in
the randomized case and $\Omega(\sqrt{\log n})$ in the quantum case for any
$\varepsilon>0$. \ We show that improving our lower bounds is intimately
related to two longstanding open problems about Boolean functions: the
Sensitivity Conjecture, and the relationships between query complexity and
polynomial degree.

\end{abstract}

\section{Introduction\label{INTRO}}

Given a Boolean input $X=\left(  x_{1},\ldots,x_{n}\right)  \in\left\{
0,1\right\}  ^{n}$, the \textsc{Parity} problem is to compute%
\begin{equation}
\text{\textsc{Par}}\left(  X\right)  :=x_{1}\oplus\cdots\oplus x_{n}\,.
\end{equation}
This is one of the most fundamental and well-studied problems in computer science.

Since \textsc{Par}$\left(  X\right)  $ is sensitive to all $n$ bits at every
input $X$, any classical\ algorithm for \textsc{Parity}\ requires examining
all $n$ bits. \ As a result, \textsc{Parity}\ is often considered a
\textquotedblleft maximally hard problem\textquotedblright\ for query or
decision-tree complexity. \ In the quantum case, one can get a \textit{slight}
improvement to $\left\lceil n/2\right\rceil $\ queries, by applying the
Deutsch-Jozsa algorithm \cite{dj} to successive pairs of coordinates ($\left(
x_{1},x_{2}\right)  $, $\left(  x_{3},x_{4}\right)  $, etc.) and then XORing
the results. \ However, that factor-of-two improvement is known to be the best
possible by quantum algorithms \cite{farhi:parity,bbcmw}.\footnote{Moreover,
this holds even for \textit{unbounded-error} quantum algorithms, which only
need to guess \textsc{Par}$\left(  X\right)  $\ with \textit{some} probability
greater than $1/2$, but must do so for every $X$.}

So we might wonder: can we learn \textit{anything} about a string's parity by
making a sublinear number of queries? \ One natural goal would be to compute
the parity, not for all inputs, but merely for as many inputs as possible.
\ This motivates the following problem, which will be the focus of this paper.

\begin{problem}
[\textsc{Weak Parity} or \textsc{WeakPar}$_{n,\varepsilon}$]Let $p\left(
X\right)  $\ be the probability that an algorithm accepts a Boolean input
$X\in\left\{  0,1\right\}  ^{n}$. \ Then given $\varepsilon>0$, satisfy%
\begin{equation}
\Pr_{X\in\left\{  0,1\right\}  ^{n}}\left[  \left\vert p\left(  X\right)
-\text{\textsc{Par}}\left(  X\right)  \right\vert \leq\frac{1}{3}\right]
\geq\frac{1}{2}+\varepsilon, \label{wp}%
\end{equation}
by querying $X$ as few times as possible. \ Equivalently, satisfy $\left\vert
A\right\vert \geq\left(  1/2+\varepsilon\right)  2^{n}$, where $A\subseteq
\left\{  0,1\right\}  ^{n}$\ is the set of all inputs $X$\ such that
$\left\vert p\left(  X\right)  -\text{\textsc{Par}}\left(  X\right)
\right\vert \leq1/3$.

We will sometimes refer to the above as \textquotedblleft
bounded-error\textquotedblright\ \textsc{Weak Parity}. \ In the
\textquotedblleft zero-error\textquotedblright\ variant, we instead want to
satisfy the stronger condition%
\begin{equation}
\Pr_{X\in\left\{  0,1\right\}  ^{n}}\left[  p\left(  X\right)
=\text{\textsc{Par}}\left(  X\right)  \right]  \geq\frac{1}{2}+\varepsilon.
\end{equation}

\end{problem}

To build intuition, let's start with some elementary remarks about
\textsc{Weak Parity}.

\begin{enumerate}
\item[(i)] Of course it's trivial to guess \textsc{Par}$\left(  X\right)
$\ on a\ $1/2$\ fraction of inputs $X$, for example by always outputting $0$.
\ (On the other hand, being \textit{wrong} on a $1/2+\varepsilon$\ fraction of
$X$'s is just as hard as being right on that fraction.)

\item[(ii)] As usual, the constant $1/3$\ in equation (\ref{wp}) is arbitrary;
we can replace it by any other constant in $\left(  0,1/2\right)  $ using amplification.

\item[(iii)] There is no requirement that the acceptance probability $p\left(
X\right)  $\ approximate a total Boolean function. \ In other words, if
$X\notin A$\ then $p\left(  X\right)  $\ can be anything in $\left[
0,1\right]  $.

\item[(iv)] It is not hard to see that \textsc{Weak Parity}\ is completely
uninteresting for deterministic classical algorithms. \ Indeed, any such
algorithm that makes fewer than $n$ queries\ correctly guesses \textsc{Par}%
$\left(  X\right)  $\ on exactly half of the inputs.

\item[(v)] Even a randomized or quantum algorithm must be \textquotedblleft
uncorrelated\textquotedblright\ with \textsc{Par}$\left(  X\right)  $, if it
always makes $T<n$\ queries (in the randomized case) or $T<n/2$\ queries (in
the quantum case). \ In other words, we must have%
\begin{equation}
\sum_{X\in\left\{  0,1\right\}  ^{n}}\left(  p\left(  X\right)  -\frac{1}%
{2}\right)  \left(  \text{\textsc{Par}}\left(  X\right)  -\frac{1}{2}\right)
=0, \label{uncor}%
\end{equation}
where $p\left(  X\right)  $\ is the algorithm's acceptance probability. \ The
reason is just Fourier analysis: if we switch domains from $\left\{
0,1\right\}  $ to $\left\{  1,-1\right\}  $, then \textsc{Par}$\left(
X\right)  =x_{1}\cdots x_{n}$. \ But for a randomized algorithm, $p\left(
X\right)  $\ is a multilinear polynomial in $x_{1},\ldots,x_{n}$\ of degree at
most $T<n$, while for a quantum algorithm, Beals et al.\ \cite{bbcmw}\ showed
that $p\left(  X\right)  $\ is a multilinear polynomial of degree at most
$2T<n$. \ And any such polynomial has correlation $0$ with the degree-$n$
monomial $x_{1}\cdots x_{n}$.

\item[(vi)] Crucially, however, equation (\ref{uncor}) does \textit{not} rule
out sublinear randomized or quantum algorithms for \textsc{Weak Parity}%
\ (which exist for all $\varepsilon=o\left(  1\right)  $, as we will see!).
\ The reason is a bit reminiscent of the famous \textit{hat puzzle}%
:\footnote{In that puzzle, $n$ players are each assigned a red hat or a blue
hat uniformly at random, and can see the colors of every hat except their own.
\ At least one player must guess the color of her own hat, and every guess
must be correct. \ Surprisingly, even though each player has only a
$1/2$\ probability of being correct, it is possible for the players to win
this game with probability $\sim1-1/n$, by \textquotedblleft
conspiring\textquotedblright\ so that the cases where they are wrong coincide
with each other. \ See http://en.wikipedia.org/wiki/Hat\_puzzle} suppose, for
example, that an algorithm output \textsc{Par}$\left(  X\right)  $\ with
probability exactly $2/3$\ on a $3/4$\ fraction of inputs $X$, and with
probability $0$ on the remaining $1/4$\ fraction of inputs. \ Such an
algorithm would succeed at \textsc{Weak Parity} for $\varepsilon=1/4$, despite
maintaining an overall correlation of $0$ with \textsc{Par}$\left(  X\right)
$.

\item[(vii)] The correlation argument does establish that, for the
\textit{zero-error} variant of \textsc{Weak Parity}, any randomized algorithm
must make at least $n$ queries, and any quantum algorithm must make at least
$n/2$\ queries, with \textit{some} nonzero probability.\footnote{For the
bounded-error variant of \textsc{Weak Parity}, the argument also establishes
that if $\varepsilon>1/4$, then any randomized algorithm\ must make
$n$\ queries, and any quantum algorithm must make $n/2$ queries.} \ Even then,
however, an algorithm that makes an \textit{expected} sublinear number of
queries on each input $X$ is not ruled out (and as we will see, such
algorithms exist).
\end{enumerate}

The regime of \textsc{Weak Parity}\ that interests us the most is where
$\varepsilon$\ is very small---the extreme case being $\varepsilon=1/2^{n}$.
\ We want to know: \textit{are there nontrivial randomized or quantum
algorithms to guess the parity of }$X$\textit{ on slightly more than half the
inputs?}

Despite an immense amount of work on query complexity, so far as we know the
above question was never asked before. \ Here we initiate its study, both by
proving upper and lower bounds, and by relating this innocent-looking question
to longstanding open problems in combinatorics, including the Sensitivity
Conjecture. \ Even though \textsc{Weak Parity}\ might look at first like a
curiosity, we will find that the task of understanding its query complexity is
tightly linked to \textit{general} questions about query complexity, and these
links help to motivate its study. \ Conversely, \textsc{Weak Parity}%
\ illustrates how an old pastime in complexity theory---namely,
understanding\ the largest possible gaps between query complexity measures for
\textit{arbitrary} Boolean functions---can actually have implications for the
query complexities of \textit{specific} problems.

\section{Our Results\label{RESULTS}}

First, in Section \ref{ALG}, we prove an upper bound of $O(n/\log
^{0.246}\left(  1/\varepsilon\right)  )$\ on the zero-error randomized query
complexity of \textsc{Weak Parity},\ and an upper bound of $O(n/\sqrt
{\log1/\varepsilon})$\ on its bounded-error quantum query complexity. \ (For
zero-error quantum query complexity, we get the slightly worse bound $O\left(
n\cdot\frac{\left(  \log\log\frac{1}{\varepsilon}\right)  ^{2}}{\sqrt
{\log1/\varepsilon}}\right)  $.)

Our quantum algorithm is based on Grover's algorithm, while our randomized
algorithm is based on the well-known $O\left(  n^{0.754}\right)  $\ randomized
algorithm for the complete binary AND/OR tree. \ For the zero-error quantum
algorithm, we use a recent zero-error quantum algorithm for the complete
binary AND/OR tree due to Ambainis et al.\ \cite{aclt}.

Then, in Section \ref{RSR}, we prove a not-quite-matching lower bound of
$\Omega\left(  n/\log\left(  1/\varepsilon\right)  \right)  $ queries, by
using random self-reducibility to reduce ordinary \textsc{Parity}\ to
\textsc{Weak Parity}. \ This lower bound is the same for randomized and
quantum, and for zero-error and bounded-error.

The gap between our upper and lower bounds might seem tiny. \ But notice that
the gap steadily worsens for smaller $\varepsilon$, reaching $O(n^{0.754}%
)$\ or $O(\sqrt{n})$\ or $O(\sqrt{n}\log^{2}n)$\ versus the trivial
$\Omega\left(  1\right)  $\ when $\varepsilon=1/2^{n}$. \ This leads us to ask
whether we can prove a nontrivial lower bound that works for \textit{all}
$\varepsilon>0$. \ Equivalently, can we rule out an $O\left(  1\right)
$-query randomized or quantum algorithm that computes \textsc{Parity} on a
subset $A\subseteq\left\{  0,1\right\}  ^{n}$\ of size $2^{n-1}+1$?

In Section \ref{SENSLB}, we show that we \textit{can} (barely) rule out such
an algorithm. \ In 1988, Chung et al.\ \cite{chung} showed that any induced
subgraph of the Boolean hypercube $\left\{  0,1\right\}  ^{n}$, of size at
least $2^{n-1}+1$, must have at least one vertex of degree $\Omega\left(  \log
n\right)  $. \ As a consequence, we deduce that for all $\varepsilon>0$, any
bounded-error randomized algorithm for \textsc{Weak Parity} must make
$\Omega(\log n)$\ queries, and any bounded-error quantum algorithm must make
$\Omega(\sqrt{\log n})$\ queries.

It has been conjectured that Chung et al.'s $\Omega\left(  \log n\right)  $
degree lower bound can be improved to $n^{\Omega\left(  1\right)  }$.
\ Previously, however, Gotsman and Linial \cite{gotsmanlinial} showed that
such an improvement would imply the notorious \textit{Sensitivity Conjecture}
in the study of Boolean functions.\ \ In Section \ref{SENSLB}, we observe that
an $n^{\Omega\left(  1\right)  }$\ lower bound for Chung et al.'s problem
would \textit{also} yield an $n^{\Omega\left(  1\right)  }$\ lower bound on
the bounded-error randomized and quantum query complexities of \textsc{Weak
Parity}, for all $\varepsilon>0$. \ Thus, while we do not have a direct
reduction between \textsc{Weak Parity} and the Sensitivity Conjecture in
either direction, it seems plausible that a breakthrough on one problem would
lead to a breakthrough on the other.

Next, in Section \ref{DQSEC}, we connect \textsc{Weak Parity}\ to another
longstanding open problem in the study of Boolean functions---and in this
case, we give a direct reduction. \ Namely, suppose we could prove a lower
bound of $\Omega\left(  n/\log^{1-c}\left(  1/\varepsilon\right)  \right)
$\ on the bounded-error randomized query complexity of \textsc{Weak Parity}.
\ We show that this would imply that $\operatorname*{R}_{2}\left(  f\right)
=\Omega\left(  \deg\left(  f\right)  ^{c}\right)  $\ for all total Boolean
functions $f:\left\{  0,1\right\}  ^{n}\rightarrow\left\{  0,1\right\}
$,\ where $\operatorname*{R}_{2}\left(  f\right)  $\ is the bounded-error
randomized query complexity of $f$, and $\deg\left(  f\right)  $\ is its exact
degree as a real polynomial. \ Similar statements hold for other kinds of
query complexity (e.g., the bounded-error quantum query complexity
$\operatorname*{Q}_{2}\left(  f\right)  $, and the zero-error randomized query
complexity $\operatorname*{R}_{0}\left(  f\right)  $).

Nisan \cite{nisan}\ showed that $\operatorname*{R}_{2}\left(  f\right)
=\Omega(\deg\left(  f\right)  ^{1/3})$ for all total Boolean functions $f$,
while Beals et al.\ \cite{bbcmw} showed that $\operatorname*{Q}_{2}\left(
f\right)  =\Omega(\deg\left(  f\right)  ^{1/6})$\ for all $f$. \ Meanwhile,
the largest known separations are $\operatorname*{R}_{2}\left(  f\right)
=O(\deg\left(  f\right)  ^{0.753\ldots})$\ if $f$ is the complete binary
AND/OR tree (see Section \ref{PRELIM}\ for a definition), and
$\operatorname*{Q}_{2}\left(  f\right)  =O(\sqrt{\deg\left(  f\right)  })$ if
$f$ is the $\operatorname*{OR}$\ function. \ However, even improving on the
$3^{rd}$- and $6^{th}$-power relations remains open. \ Our result says that,
if there existed Boolean functions $f$ with larger separations than are
currently known, then we \textit{could} improve our algorithms for
\textsc{Weak Parity}. \ And conversely, any randomized lower bound for
\textsc{Weak Parity}\ better than $\Omega(n/\log^{2/3}\left(  1/\varepsilon
\right)  )$, or any quantum lower bound\ better than $\Omega(n/\log
^{5/6}\left(  1/\varepsilon\right)  )$, would improve the known relations
between degree and query complexity for \textit{all} Boolean functions.

Lastly, in Section \ref{OTHER}, we briefly consider the weak query
complexities of functions other than \textsc{Parity}. \ We show that, for
\textit{every} Boolean function $f$, it is possible to agree with $f\left(
X\right)  $\ on $2^{n-1}+1$\ inputs $X$\ using a bounded-error quantum
algorithm that makes $O(\sqrt{n})$\ queries, or a zero-error randomized
algorithm that makes $O(n^{0.754})$\ queries, or a zero-error quantum
algorithm that makes $O(\sqrt{n}\log^{2}n)$\ queries.

\section{Preliminaries\label{PRELIM}}

We assume some familiarity with classical and quantum query complexity; see
Buhrman and de Wolf \cite{bw} for an excellent introduction. \ This section
reviews the most relevant definitions and facts.

\subsection{Classical Query Complexity\label{QUERY}}

Given a Boolean function $f:\left\{  0,1\right\}  ^{n}\rightarrow\left\{
0,1\right\}  $, the \textit{deterministic query complexity} $\operatorname*{D}%
\left(  f\right)  $ is the minimum number of queries made by any
deterministic, classical algorithm that computes $f\left(  X\right)  $\ for
every input $X\in\left\{  0,1\right\}  ^{n}$. \ (Here and throughout, a query
returns $x_{i}$\ given $i$, and the \textquotedblleft number of
queries\textquotedblright\ means the number maximized over all $X\in\left\{
0,1\right\}  ^{n}$.)

Also, the \textit{zero-error randomized query complexity} $\operatorname*{R}%
_{0}\left(  f\right)  $\ is the minimum number of queries made by any
randomized algorithm that computes $f\left(  X\right)  $\ with success
probability at least $2/3$\ for every $X$---and that, whenever it fails to
compute $f\left(  X\right)  $, instead outputs \textquotedblleft don't
know.\textquotedblright\ \ The \textit{bounded-error randomized query
complexity} $\operatorname*{R}_{2}\left(  f\right)  $\ is the minimum number
of queries made by a randomized algorithm that computes $f\left(  X\right)
$\ with success probability at least $2/3$ for every $X$, and that can behave
arbitrarily (for example, by outputting the wrong answer) when it fails. \ We
have the following relations for every $f$:%
\begin{equation}
n\geq\operatorname*{D}\left(  f\right)  \geq\operatorname*{R}\nolimits_{0}%
\left(  f\right)  \geq\operatorname*{R}\nolimits_{2}\left(  f\right)  .
\end{equation}
We could also have defined $\operatorname*{R}\nolimits_{0}\left(  f\right)
$\ as the minimum \textit{expected} number of queries made by a randomized
algorithm that computes $f\left(  X\right)  $\ with certainty for every input
$X$ (where the expectation is over the internal randomness of the algorithm,
and must be bounded for every $X$). \ We will sometimes use this
interpretation, which changes the value of $\operatorname*{R}\nolimits_{0}%
\left(  f\right)  $\ by at most a constant factor.

We will use the following well-known result:

\begin{theorem}
\label{drthm}$\operatorname*{D}\left(  f\right)  \leq\operatorname*{R}%
_{0}\left(  f\right)  ^{2}$ and $\operatorname*{D}\left(  f\right)
=O(\operatorname*{R}_{2}\left(  f\right)  ^{3})$ for all total Boolean
functions $f$.\footnote{The $\operatorname*{D}\left(  f\right)  \leq
\operatorname*{R}_{0}\left(  f\right)  ^{2}$\ part follows from the folklore
result that $\operatorname*{D}\left(  f\right)  \leq\operatorname*{C}\left(
f\right)  ^{2}$, where $\operatorname*{C}\left(  f\right)  $\ is the so-called
\textit{certificate complexity}, together with the fact that
$\operatorname*{R}_{0}\left(  f\right)  \geq\operatorname*{C}\left(  f\right)
$. \ The $\operatorname*{D}\left(  f\right)  =O(\operatorname*{R}_{2}\left(
f\right)  ^{3})$\ part was proved by Nisan \cite{nisan}. \ It also follows
from the result of Beals et al.\ \cite{bbcmw} that $\operatorname*{D}\left(
f\right)  \leq\operatorname*{bs}\left(  f\right)  ^{3}$, where
$\operatorname*{bs}\left(  f\right)  $\ is the \textit{block sensitivity} (see
Section \ref{SBS}), together with the fact that $\operatorname*{R}_{2}\left(
f\right)  =\Omega\left(  \operatorname*{bs}\left(  f\right)  \right)  $.}
\end{theorem}

We will write $\operatorname*{R}_{2}($\textsc{WeakPar}$_{n,\varepsilon})$\ to
denote the minimum number of queries made by any randomized algorithm that,
for at least a $1/2+\varepsilon$\ fraction of inputs $X\in\left\{
0,1\right\}  ^{n}$, outputs \textsc{Par}$\left(  X\right)  $\ with probability
at least $2/3$. \ We will also write $\operatorname*{R}_{0}($\textsc{WeakPar}%
$_{n,\varepsilon})$\ to denote the minimum number of queries made by any
randomized algorithm that, for at least a $1/2+\varepsilon$\ fraction of
inputs $X$, outputs \textsc{Par}$\left(  X\right)  $\ with probability at
least $2/3$, and otherwise outputs \textquotedblleft don't
know.\textquotedblright\ \ In both cases, for the remaining inputs $X$ (i.e.,
those on which the algorithm fails), the algorithm's output behavior can be
arbitrary, but the upper bound on query complexity must hold for \textit{all}
inputs $X\in\left\{  0,1\right\}  ^{n}$.

Note that we could also define $\operatorname*{R}_{0}^{\prime}($%
\textsc{WeakPar}$_{n,\varepsilon})$\ as the minimum \textit{expected} number
of queries made by any randomized algorithm that, for at least a
$1/2+\varepsilon$\ fraction of inputs $X$, outputs \textsc{Par}$\left(
X\right)  $\ with probability $1$. \ In this case, the expected number of
queries needs to be bounded only for those $X$'s on which the algorithm
succeeds. \ For completeness, let us verify the following.

\begin{proposition}
\label{r0r0}$\operatorname*{R}_{0}($\textsc{WeakPar}$_{n,\varepsilon})$ and
$\operatorname*{R}_{0}^{\prime}($\textsc{WeakPar}$_{n,\varepsilon})$\ are
equal up to constant factors.
\end{proposition}

\begin{proof}
Let $A$ be a randomized algorithm that realizes $\operatorname*{R}_{0}%
($\textsc{WeakPar}$_{n,\varepsilon})\leq T$. \ Then we can simply run $A$
repeatedly, until it outputs either $0$ or $1$. \ This will yield an algorithm
that, for at least a $1/2+\varepsilon$\ fraction of inputs $X\in\left\{
0,1\right\}  ^{n}$, outputs \textsc{Par}$\left(  X\right)  $\ with certainty
after $O\left(  T\right)  $\ queries in expectation. \ (The algorithm might
not halt for the remaining $X$'s, but that's okay.)

Conversely, let $A^{\prime}$\ be a randomized algorithm that realizes
$\operatorname*{R}_{0}^{\prime}($\textsc{WeakPar}$_{n,\varepsilon})\leq T$.
\ Then we can run $A^{\prime}$\ until it's either halted or made
$3T$\ queries, and can output \textquotedblleft don't know\textquotedblright%
\ in the latter case. \ By Markov's inequality, this will yield an algorithm
that, for at least a $1/2+\varepsilon$\ fraction of inputs $X$, outputs
\textsc{Par}$\left(  X\right)  $\ with probability at least $2/3$, and
otherwise outputs \textquotedblleft don't know.\textquotedblright%
\ \ Furthermore, the number of queries will be bounded by $3T$\ for every $X$.
\end{proof}

\subsection{Quantum Query Complexity\label{QQUERY}}

The\ \textit{zero-error quantum query complexity} $\operatorname*{Q}%
_{0}\left(  f\right)  $\ is the minimum number of queries made by any quantum
algorithm that computes $f\left(  X\right)  $\ with success probability at
least $2/3$, for every input $X$---and that, whenever it fails to compute
$f\left(  X\right)  $, instead outputs \textquotedblleft don't
know.\textquotedblright\ \ Here a query maps each computational basis state of
the form $\left\vert i,b,z\right\rangle $ to a basis state of the form
$\left\vert i,b\oplus x_{i},z\right\rangle $, where $z$ is a \textquotedblleft
workspace register\textquotedblright\ whose dimension can be arbitrary. \ The
final output ($0$, $1$, or \textquotedblleft don't know\textquotedblright) is
obtained by measuring a designated part of $z$. \ The \textit{bounded-error
randomized query complexity} $\operatorname*{Q}_{2}\left(  f\right)  $\ is the
minimum number of queries made by a quantum algorithm that computes $f\left(
X\right)  $\ with success probability at least $2/3$ for every $X$, and whose
output can be arbitrary when it fails. \ We have the following relations for
every $f$:%
\begin{equation}
\operatorname*{R}\nolimits_{0}\left(  f\right)  \geq\operatorname*{Q}%
\nolimits_{0}\left(  f\right)  \geq\operatorname*{Q}\nolimits_{2}\left(
f\right)  ,~~~~\operatorname*{R}\nolimits_{2}\left(  f\right)  \geq
\operatorname*{Q}\nolimits_{2}\left(  f\right)  .
\end{equation}
Like in the randomized case, we can also interpret $\operatorname*{Q}%
\nolimits_{0}\left(  f\right)  $\ as the minimum \textit{expected} number of
queries made by a quantum algorithm that computes $f\left(  X\right)  $\ with
certainty for every input $X$, if we generalize the quantum query model to
allow intermediate measurements. \ Doing so changes $\operatorname*{Q}%
\nolimits_{0}\left(  f\right)  $\ by at most a constant factor.

We will use the following results of Beals et al.\ \cite{bbcmw}\ and
Midrijanis \cite{midrijanis:drq} respectively:

\begin{theorem}
[Beals et al.\ \cite{bbcmw}]\label{bbcmwthm}$\operatorname*{D}\left(
f\right)  =O(\operatorname*{Q}_{2}\left(  f\right)  ^{6})$ for all total
Boolean $f$.
\end{theorem}

\begin{theorem}
[Midrijanis \cite{midrijanis:drq}]\label{midrijanisthm}$\operatorname*{D}%
\left(  f\right)  =O(\operatorname*{Q}_{0}\left(  f\right)  ^{3})$\ for all
total Boolean $f$.\footnote{This improved the result of Buhrman et
al.\ \cite{bcwz}\ that $\operatorname*{D}\left(  f\right)
=O(\operatorname*{Q}_{0}\left(  f\right)  ^{4})$, as well as the result of
Aaronson \cite{aar:cer}\ that $\operatorname*{R}\nolimits_{0}\left(  f\right)
=O(\operatorname*{Q}_{0}\left(  f\right)  ^{3}\log n)$.}
\end{theorem}

Just like in the randomized case, we will write $\operatorname*{Q}_{2}%
($\textsc{WeakPar}$_{n,\varepsilon})$ for the minimum number of queries made
by any quantum algorithm that, for at least a $1/2+\varepsilon$\ fraction of
inputs $X$, outputs \textsc{Par}$\left(  X\right)  $\ with probability at
least $2/3$; and will write $\operatorname*{Q}_{0}($\textsc{WeakPar}%
$_{n,\varepsilon})$\ for the minimum number of queries made by any quantum
algorithm that, for at least a $1/2+\varepsilon$\ fraction of $X$'s, outputs
\textsc{Par}$\left(  X\right)  $\ with probability at least $2/3$, and
otherwise outputs \textquotedblleft don't know.\textquotedblright

Once again, if we generalize the quantum query model to allow intermediate
measurements, then we can also define $\operatorname*{Q}_{0}($\textsc{WeakPar}%
$_{n,\varepsilon})$\ as the minimum \textit{expected} number of queries made
by any quantum algorithm that, for at least a $1/2+\varepsilon$\ fraction of
$X$'s, outputs \textsc{Par}$\left(  X\right)  $\ with probability $1$ (with
the expected number of queries bounded only for those $X$'s on which the
algorithm succeeds). \ Doing so changes $\operatorname*{Q}_{0}($%
\textsc{WeakPar}$_{n,\varepsilon})$\ by at most a constant factor, for the
same reasons as in Proposition \ref{r0r0}.

\subsection{Degree\label{DEGREE}}

Given a Boolean function $f$, the \textit{degree} $\deg\left(  f\right)  $\ is
the degree of the (unique) real multilinear polynomial $p:\mathbb{R}%
^{n}\rightarrow\mathbb{R}$\ that satisfies $p\left(  X\right)  =f\left(
X\right)  $\ for all $X\in\left\{  0,1\right\}  ^{n}$. \ Degree has a known
combinatorial characterization that will be useful to us:\footnote{For a proof
of this characterization, see for example Aaronson \cite{aar:bf}.}

\begin{proposition}
[folklore]\label{degreeprop}Given a $d$-dimensional subcube $S$\ in $\left\{
0,1\right\}  ^{n}$, let $S_{0},S_{1}$\ be the subsets of $S$\ with even and
odd Hamming weight respectively (thus $\left\vert S_{0}\right\vert =\left\vert
S_{1}\right\vert =2^{d-1}$). \ Also, given a Boolean function $f:\left\{
0,1\right\}  ^{n}\rightarrow\left\{  0,1\right\}  $, call $f$
\textquotedblleft parity-correlated\textquotedblright\ on $S$\ if%
\begin{equation}
\left\vert \left\{  X\in S_{0}:f\left(  X\right)  =1\right\}  \right\vert
\neq\left\vert \left\{  X\in S_{1}:f\left(  X\right)  =1\right\}  \right\vert
.
\end{equation}
Then $\deg\left(  f\right)  $\ equals the maximum dimension of a subcube on
which $f$ is parity-correlated.
\end{proposition}

It is not hard to see that $\deg\left(  f\right)  \leq\operatorname*{D}\left(
f\right)  $ for all Boolean functions $f$. \ Combined with Theorems
\ref{drthm} and \ref{bbcmwthm}, this\ implies that $\operatorname*{R}%
_{2}\left(  f\right)  =\Omega(\deg\left(  f\right)  ^{1/3})$\ and
$\operatorname*{Q}_{2}\left(  f\right)  =\Omega(\deg\left(  f\right)  ^{1/6}%
)$, as stated in Section \ref{RESULTS}.

\subsection{Sensitivity and Block Sensitivity\label{SBS}}

Given an input $X\in\left\{  0,1\right\}  ^{n}$ and a subset $B\subseteq
\left[  n\right]  $, let $X^{B}$\ denote $X$\ with all the bits in $B$
flipped. \ Then for a Boolean function $f$, the \textit{sensitivity}
$\operatorname*{s}^{X}\left(  f\right)  $\ is the number of indices
$i\in\left[  n\right]  $\ such that $f\left(  X^{\left\{  i\right\}  }\right)
\neq f\left(  X\right)  $, while the \textit{block sensitivity}
$\operatorname*{bs}^{X}\left(  f\right)  $\ is the maximum number of
pairwise-disjoint \textquotedblleft blocks\textquotedblright\ $B_{1}%
,\ldots,B_{k}\subseteq\left[  n\right]  $\ that can be found such that
$f\left(  X^{B_{j}}\right)  \neq f\left(  X\right)  $\ for all $j\in\left[
k\right]  $. \ We then define%
\begin{equation}
\operatorname*{s}\left(  f\right)  :=\max_{X\in\left\{  0,1\right\}  ^{n}%
}\operatorname*{s}\nolimits^{X}\left(  f\right)  ,~~~~~\operatorname*{bs}%
\left(  f\right)  :=\max_{X\in\left\{  0,1\right\}  ^{n}}\operatorname*{bs}%
\nolimits^{X}\left(  f\right)  .
\end{equation}
Clearly $\operatorname*{s}\left(  f\right)  \leq\operatorname*{bs}\left(
f\right)  $. \ The famous \textit{Sensitivity Conjecture} (see Hatami et
al.\ \cite{hatami}\ for a survey)\ asserts that the gap between
$\operatorname*{s}\left(  f\right)  $\ and $\operatorname*{bs}\left(
f\right)  $\ is never more than polynomial:\footnote{Rubinstein
\cite{rubinstein} showed that $\operatorname*{bs}\left(  f\right)
$\ \textit{can} be quadratically larger than $\operatorname*{s}\left(
f\right)  $.}

\begin{conjecture}
[Sensitivity Conjecture]\label{sensconj}There exists a polynomial $p$ such
that $\operatorname*{bs}\left(  f\right)  \leq p\left(  \operatorname*{s}%
\left(  f\right)  \right)  $ for all $f$.
\end{conjecture}

Nisan and Szegedy \cite{ns}\ showed that $\operatorname*{bs}\left(  f\right)
\leq2\deg\left(  f\right)  ^{2}$ (recently improved by Tal \cite{tal}\ to
$\operatorname*{bs}\left(  f\right)  \leq\deg\left(  f\right)  ^{2}$), while
Beals et al.\ \cite{bbcmw} showed that $\deg\left(  f\right)  \leq
\operatorname*{bs}\left(  f\right)  ^{3}$.\footnote{This follows immediately
from their result that $\operatorname*{D}\left(  f\right)  \leq
\operatorname*{bs}\left(  f\right)  ^{3}$, which improved on the bound
$\operatorname*{D}\left(  f\right)  \leq\operatorname*{bs}\left(  f\right)
^{4}$\ due to Nisan and Szegedy \cite{ns}, and which they then combined with
the result $\operatorname*{Q}_{2}\left(  f\right)  =\Omega(\sqrt
{\operatorname*{bs}\left(  f\right)  })$\ to prove Theorem \ref{bbcmwthm},
that $\operatorname*{D}\left(  f\right)  =O(\operatorname*{Q}_{2}\left(
f\right)  ^{6})$.} \ Thus, degree and block sensitivity are polynomially
related. \ This implies that Conjecture \ref{sensconj} is equivalent to the
conjecture that sensitivity\ is polynomially related to degree.

\subsection{AND/OR Tree\label{AOT}}

A particular Boolean function of interest to us will be the \textit{complete
binary AND/OR tree}. \ Assume $n=2^{d}$; then this function is defined
recursively as follows:%
\begin{align}
\operatorname*{T}\nolimits_{0}\left(  x\right)   &  :=x,\\
\operatorname*{T}\nolimits_{d}\left(  x_{1},\ldots,x_{n}\right)   &
:=\left\{
\begin{array}
[c]{cc}%
\operatorname*{T}\nolimits_{d-1}\left(  x_{1},\ldots,x_{n/2}\right)
\operatorname*{AND}\operatorname*{T}\nolimits_{d-1}\left(  x_{n/2+1}%
,\ldots,x_{n}\right)  & \text{if }d>0\text{ is odd,}\\
\operatorname*{T}\nolimits_{d-1}\left(  x_{1},\ldots,x_{n/2}\right)
\operatorname*{OR}\operatorname*{T}\nolimits_{d-1}\left(  x_{n/2+1}%
,\ldots,x_{n}\right)  & \text{if }d>0\text{ is even.}%
\end{array}
\right.
\end{align}
It is not hard to see that%
\begin{equation}
\operatorname*{D}\left(  \operatorname*{T}\nolimits_{d}\right)  =\deg\left(
\operatorname*{T}\nolimits_{d}\right)  =2^{d}=n.
\end{equation}
By contrast, Saks and Wigderson \cite{sakswigderson} proved the following.

\begin{theorem}
[Saks-Wigderson \cite{sakswigderson}]\label{swthm}$\operatorname*{R}%
_{0}\left(  \operatorname*{T}_{d}\right)  =O\left(  \left(  \frac{1+\sqrt{33}%
}{4}\right)  ^{d}\right)  =O(n^{0.753\ldots})$.
\end{theorem}

Saks and Wigderson \cite{sakswigderson}\ also proved a matching lower bound of
$\operatorname*{R}_{0}\left(  \operatorname*{T}_{d}\right)  =\Omega
(n^{0.753\ldots})$, while Santha \cite{santha}\ proved that $\operatorname*{R}%
_{2}\left(  \operatorname*{T}_{d}\right)  =\Omega(n^{0.753\ldots})$ even for
bounded-error algorithms. \ Note that $\operatorname*{T}_{d}$\ gives the
largest known gap between $\operatorname*{D}\left(  f\right)  $\ and
$\operatorname*{R}_{2}\left(  f\right)  $\ for any total Boolean function $f$.

Recently, building on the breakthrough quantum walk algorithm for game-tree
evaluation \cite{fgg:nand} (see also \cite{acrsz}), Ambainis et
al.\ \cite{aclt} proved the following.

\begin{theorem}
[Ambainis et al.\ \cite{aclt}]\label{acltthm}$\operatorname*{Q}_{0}\left(
\operatorname*{T}_{d}\right)  =O(\sqrt{n}\log^{2}n).$
\end{theorem}

By comparison, it is not hard to show (by reduction from \textsc{Parity}) that
$\operatorname*{Q}_{2}\left(  \operatorname*{T}_{d}\right)  =\Omega(\sqrt{n}%
)$. \ Once again, Theorem \ref{acltthm}\ gives the largest known gap between
$\operatorname*{D}\left(  f\right)  $\ and $\operatorname*{Q}_{0}\left(
f\right)  $\ for any total $f$.\footnote{It improves slightly on an earlier
result of Buhrman et al.\ \cite{bcwz}, who showed that for every
$\varepsilon>0$, there exists an $f$ such that $\operatorname*{Q}_{0}\left(
f\right)  =O(\operatorname*{D}\left(  f\right)  ^{1/2+\varepsilon})$. \ For
$\operatorname*{Q}_{2}$, we can do slightly better ($\operatorname*{Q}%
_{2}\left(  f\right)  =O(\sqrt{\operatorname*{D}\left(  f\right)  })$) by just
taking $f$ to be the $\operatorname*{OR}$\ function.}

Finally, the following fact will be useful to us.

\begin{proposition}
\label{tvspar}Let $n=2^{d}$. \ The number of inputs $X\in\left\{  0,1\right\}
^{n}$\ such that $\operatorname*{T}_{d}\left(  X\right)  ={}$\textsc{Par}%
$\left(  X\right)  $\ is exactly $2^{n-1}+1$\ if $d$ is even, and exactly
$2^{n-1}-1$\ if $d$ is odd.
\end{proposition}

\begin{proof}
This is most easily proved by switching to the Fourier representation. \ Let%
\begin{align}
\operatorname*{T}\nolimits_{d}^{\ast}\left(  x_{1},\ldots,x_{n}\right)   &
:=1-2\operatorname*{T}\nolimits_{d}\left(  \frac{1-x_{1}}{2},\ldots
,\frac{1-x_{n}}{2}\right)  ,\\
\text{\textsc{Par}}^{\ast}\left(  X\right)   &  :=x_{1}\cdots x_{n}.
\end{align}
Then the problem reduces to computing the correlation%
\begin{equation}
C_{d}:=\sum_{X\in\left\{  0,1\right\}  ^{n}}\operatorname*{T}\nolimits_{d}%
^{\ast}\left(  X\right)  \text{\textsc{Par}}^{\ast}\left(  X\right)  ,
\end{equation}
since%
\begin{equation}
\left\vert \left\{  X\in\left\{  0,1\right\}  ^{n}:\operatorname*{T}%
\nolimits_{d}\left(  X\right)  =\text{\textsc{Par}}\left(  X\right)  \right\}
\right\vert =2^{n-1}+\frac{C_{d}}{2}.
\end{equation}
We claim, by induction on $d$, that $C_{d}=2$\ if $d$\ is even, and $C_{d}%
=-2$\ if $d$ is odd. \ Certainly this holds for the base case $d=0$. \ For
larger $d$, using the fact that every two distinct monomials have correlation
$0$, one can check by calculation that%
\begin{equation}
C_{d}=\left\{
\begin{array}
[c]{cc}%
-C_{d-1}^{2}/2 & \text{if }d\text{ is odd,}\\
C_{d-1}^{2}/2 & \text{if }d\text{ is even.}%
\end{array}
\right.
\end{equation}

\end{proof}

\section{Algorithms for \textsc{Weak Parity}\label{ALG}}

We now prove our first result: that there exist nontrivial randomized and
quantum algorithms for \textsc{Weak Parity}. \ For simplicity, we first
consider the special case $\varepsilon=2^{-n}$; later we will generalize to
arbitrary $\varepsilon$.

\begin{lemma}
\label{specialcase}We have%
\begin{align}
\operatorname*{Q}\nolimits_{2}(\text{\textsc{WeakPar}}_{n,2^{-n}})  &
=O(\sqrt{n}),\\
\operatorname*{R}\nolimits_{0}(\text{\textsc{WeakPar}}_{n,2^{-n}})  &
=O(n^{0.754}),\\
\operatorname*{Q}\nolimits_{0}(\text{\textsc{WeakPar}}_{n,2^{-n}})  &
=O(\sqrt{n}\log^{2}n).
\end{align}

\end{lemma}

\begin{proof}
For $\operatorname*{Q}\nolimits_{2}$, observe that the $\operatorname*{OR}%
$\ function, $\operatorname*{OR}\left(  X\right)  $, agrees with the parity of
$X$ on $2^{n-1}+1$\ inputs $X\in\left\{  0,1\right\}  ^{n}$: namely, all the
inputs of odd Hamming weight, plus the input $0^{n}$. \ Thus, simply computing
$\operatorname*{OR}\left(  X\right)  $\ gives us an algorithm for
\textsc{WeakPar}$_{n,\varepsilon}$ with $\varepsilon=2^{-n}$. \ And of course,
$\operatorname*{OR}$\ can be computed with bounded error in $O\left(  \sqrt
{n}\right)  $\ quantum queries, using Grover's algorithm.

For $\operatorname*{R}\nolimits_{0}$, assume for simplicity that $n$ has the
form $2^{d}$; this will not affect the asymptotics. \ By Proposition
\ref{tvspar}, if $d$ is even then the AND/OR tree\ $\operatorname*{T}%
_{d}\left(  X\right)  $\ agrees with \textsc{Par}$\left(  X\right)  $\ on
$2^{n-1}+1$\ inputs $X$, while if $d$ is odd then $1-\operatorname*{T}%
_{d}\left(  X\right)  $\ does. \ Either way, simply computing
$\operatorname*{T}_{d}\left(  X\right)  $\ gives us an algorithm for
\textsc{WeakPar}$_{n,2^{-n}}$. \ Furthermore, by Theorem \ref{swthm}, there is
a zero-error randomized algorithm for $\operatorname*{T}_{d}\left(  X\right)
$\ that makes $O(n^{0.754})$\ queries.

For $\operatorname*{Q}\nolimits_{0}$, we also compute either
$\operatorname*{T}_{d}\left(  X\right)  $\ or $1-\operatorname*{T}_{d}\left(
X\right)  $\ as our guess for \textsc{Par}$\left(  X\right)  $,\ except now we
use the zero-error quantum algorithm of Theorem \ref{acltthm}, which makes
$O(\sqrt{n}\log^{2}n)$\ queries.
\end{proof}

Next, we give a general strategy for converting a \textsc{Weak Parity}%
\ algorithm for small $\varepsilon$ into an algorithm that works for larger
$\varepsilon$, with the query complexity gradually increasing as $\varepsilon$\ does.

\begin{lemma}
\label{epsrelation}For all positive integers $k$, we have%
\begin{equation}
\operatorname*{R}\nolimits_{2}(\text{\textsc{WeakPar}}_{kn,\varepsilon})\leq
k\cdot\operatorname*{R}\nolimits_{2}(\text{\textsc{WeakPar}}_{n,\varepsilon}).
\label{ineq1}%
\end{equation}
So in particular, suppose $\operatorname*{R}\nolimits_{2}($\textsc{WeakPar}%
$_{n,1/f\left(  n\right)  })\leq T\left(  n\right)  $. \ Then for all $N$ and
$\varepsilon>0$,%
\begin{equation}
\operatorname*{R}\nolimits_{2}(\text{\textsc{WeakPar}}_{N,\varepsilon}%
)\leq\frac{N\cdot T\left(  f^{-1}\left(  1/\varepsilon\right)  \right)
}{f^{-1}\left(  1/\varepsilon\right)  }. \label{ineq2}%
\end{equation}
Exactly the same holds if we replace $\operatorname*{R}\nolimits_{2}$\ by
$\operatorname*{R}\nolimits_{0}$,\ $\operatorname*{Q}_{2}$, or
$\operatorname*{Q}_{0}$ throughout.
\end{lemma}

\begin{proof}
Let $A$ be a randomized algorithm for \textsc{WeakPar}$_{n,\varepsilon}$, and
let $X$ be an input to \textsc{WeakPar}\ of size $kn$. \ Then our strategy is
to group the bits of $X$ into $n$ blocks $Y_{1},\ldots,Y_{n}$\ of $k$ bits
each, then run $A$ on the input%
\begin{equation}
\text{\textsc{Par}}\left(  Y_{1}\right)  ,\ldots,\text{\textsc{Par}}\left(
Y_{n}\right)  ,
\end{equation}
and output whatever $A$ outputs.\ \ If $A$ made $T\left(  n\right)  $\ queries
originally, then this strategy can be implemented using $k\cdot T\left(
n\right)  $ queries: namely, $k$ queries to the underlying input $X$\ every
time $A$\ queries a bit \textsc{Par}$\left(  Y_{i}\right)  $. \ Furthermore,
let $p\left(  Z\right)  $\ be $A$'s success probability on input $Z\in\left\{
0,1\right\}  ^{n}$. \ Then the strategy succeeds whenever%
\begin{equation}
\left\vert p\left(  \text{\textsc{Par}}\left(  Y_{1}\right)  ,\ldots
,\text{\textsc{Par}}\left(  Y_{n}\right)  \right)  -\left(  \text{\textsc{Par}%
}\left(  Y_{1}\right)  \oplus\cdots\oplus\text{\textsc{Par}}\left(
Y_{n}\right)  \right)  \right\vert \leq\frac{1}{3},
\end{equation}
and by assumption, this occurs for at least a $1/2+\varepsilon$\ fraction of
$Z$'s.

The inequality (\ref{ineq2}) is just a rewriting of (\ref{ineq1}), if we make
the substitutions $\varepsilon:=1/f\left(  n\right)  $\ and $n:=f^{-1}\left(
1/\varepsilon\right)  $ to get%
\begin{equation}
\operatorname*{R}\nolimits_{2}(\text{\textsc{WeakPar}}_{f^{-1}\left(
1/\varepsilon\right)  ,\varepsilon})\leq T\left(  f^{-1}\left(  1/\varepsilon
\right)  \right)  ,
\end{equation}
followed by $k:=N/f^{-1}\left(  1/\varepsilon\right)  $. \ Finally, since we
never used that $A$\ was classical or bounded-error, everything in the proof
still works if we replace $\operatorname*{R}\nolimits_{2}$\ by
$\operatorname*{R}\nolimits_{0}$,\ $\operatorname*{Q}_{2}$, or
$\operatorname*{Q}_{0}$ throughout.
\end{proof}

Combining Lemmas \ref{specialcase}\ and \ref{epsrelation}\ now easily gives us
our upper bounds:

\begin{theorem}
\label{upper}For all $n$ and $\varepsilon\in\left[  2^{-n},1/2\right]  $, we
have%
\begin{align}
\operatorname*{Q}\nolimits_{2}(\text{\textsc{WeakPar}}_{n,\varepsilon})  &
=O\left(  \frac{n}{\sqrt{\log1/\varepsilon}}\right)  ,\\
\operatorname*{R}\nolimits_{0}(\text{\textsc{WeakPar}}_{n,\varepsilon})  &
=O\left(  \frac{n}{\log^{0.246}1/\varepsilon}\right)  ,\\
\operatorname*{Q}\nolimits_{0}(\text{\textsc{WeakPar}}_{n,\varepsilon})  &
=O\left(  n\cdot\frac{\left(  \log\log1/\varepsilon\right)  ^{2}}{\sqrt
{\log1/\varepsilon}}\right)  .
\end{align}

\end{theorem}

We do not know any upper bound on $\operatorname*{R}\nolimits_{2}%
($\textsc{WeakPar}$_{n,\varepsilon})$\ better than our upper bound on
$\operatorname*{R}\nolimits_{0}($\textsc{WeakPar}$_{n,\varepsilon})$.

As a final note, all of our algorithms actually satisfy a stronger property
than the definition of \textsc{Weak Parity}\ requires. \ Namely, the
algorithms all compute a total Boolean function $f\left(  X\right)  $\ that
agrees with \textsc{Par}$\left(  X\right)  $\ on a $1/2+\varepsilon$\ fraction
of inputs. \ This means, for example, that we can obtain a randomized
algorithm that outputs \textsc{Par}$\left(  X\right)  $\ with probability $1$
on a $1/2+\varepsilon$\ fraction of inputs\ $X\in\left\{  0,1\right\}  ^{n}$,
and that halts after $O(n/\log^{0.246}\left(  1/\varepsilon\right)
)$\ queries in expectation on \textit{every} input $X$\ (not just those inputs
for which the algorithm succeeds). \ We can similarly obtain a quantum
algorithm with expected query complexity\ $O\left(  n\cdot\frac{\left(
\log\log1/\varepsilon\right)  ^{2}}{\sqrt{\log1/\varepsilon}}\right)  $\ and
the same success condition.

\section{Lower Bound via Random Self-Reducibility\label{RSR}}

Our next result is a \textit{lower} bound on the bounded-error randomized and
quantum query complexities of \textsc{Weak Parity}. \ The lower bound matches
our upper bounds in its dependence on $n$, though not in its dependence on
$\varepsilon$.

\begin{theorem}
\label{lower}$\operatorname*{Q}_{2}($\textsc{WeakPar}$_{n,\varepsilon}%
)=\Omega\left(  n/\log\left(  1/\varepsilon\right)  \right)  $ for all
$0<\varepsilon<\frac{1}{2}$.
\end{theorem}

\begin{proof}
Let $C$\ be a quantum algorithm for \textsc{WeakPar}$_{n,\varepsilon}$ that
never makes more than $T$ queries. \ Using $C$, we will produce a new quantum
algorithm $C^{\prime}$, which makes $O\left(  T\log\frac{1}{\varepsilon
}\right)  $\ queries, and which guesses \textsc{Par}$\left(  X\right)  $\ on
\textit{every} input $X\in\left\{  0,1\right\}  ^{n}$\ with probability
stricter greater than $1/2$.\ \ But it is well-known that any quantum
algorithm of the latter kind must make at least $n/2$\ queries: in other
words, that \textsc{Parity}\ has unbounded-error quantum query complexity
$n/2$\ (this follows from the polynomial method \cite{bbcmw}). \ Putting the
two facts together, we conclude that%
\begin{equation}
T=\Omega\left(  \frac{n}{\log1/\varepsilon}\right)  .
\end{equation}

To produce $C^{\prime}$, the first step is simply to amplify $C$. \ Thus, let
$C^{\ast}$ be an algorithm that outputs the majority answer among
$d\log1/\varepsilon$\ invocations of $C$. \ Then by a Chernoff bound, provided
the constant $d$ is sufficiently large,%
\begin{equation}
\Pr_{X\in\left\{  0,1\right\}  ^{n}}\left[  \left\vert \Pr\left[  C^{\ast
}\left(  X\right)  \text{ accepts}\right]  -\text{\textsc{Par}}\left(
X\right)  \right\vert \leq\varepsilon\right]  \geq\frac{1}{2}+\varepsilon.
\end{equation}
Next, $C^{\prime}$\ chooses a string $Y\in\left\{  0,1\right\}  ^{n}%
$\ uniformly at random and sets $Z:=X\oplus Y$. \ It then runs $C^{\ast}$\ to
obtain a guess $b$ about \textsc{Par}$\left(  Z\right)  $. \ Finally,
$C^{\prime}$ outputs \textsc{Par}$\left(  Y\right)  \oplus b$\ as its guess
for \textsc{Par}$\left(  X\right)  $.

Clearly $C^{\prime}$\ has the same quantum query complexity as $C^{\ast}$: it
is easy to simulate a query to a bit $z_{i}$ of $Z$, by querying the
corresponding bit $x_{i}$\ of $X$\ and then XORing with $y_{i}$.
\ Furthermore, notice that $Z$\ is uniformly random, regardless of $X$, and
that if $b=\text{\textsc{Par}}\left(  Z\right)  $\ then \textsc{Par}$\left(
Y\right)  \oplus b=$\textsc{Par}$\left(  X\right)  $. \ It follows
that\ $C^{\prime}$\ succeeds with probability at least%
\begin{equation}
\left(  \frac{1}{2}+\varepsilon\right)  \left(  1-\varepsilon\right)
=\frac{1}{2}+\frac{\varepsilon}{2}-\varepsilon^{2}>\frac{1}{2}%
\end{equation}
for every $X$, which is what we wanted to show.
\end{proof}

Of course, Theorem \ref{lower}\ implies that $\operatorname*{Q}_{0}%
($\textsc{WeakPar}$_{n,\varepsilon})$, $\operatorname*{R}_{2}($%
\textsc{WeakPar}$_{n,\varepsilon})$, and $\operatorname*{R}_{0}($%
\textsc{WeakPar}$_{n,\varepsilon})$ are $\Omega\left(  n/\log\left(
1/\varepsilon\right)  \right)  $\ as well. \ It is curious that we do not get
any lower bounds for $\operatorname*{Q}_{0}$,\ $\operatorname*{R}_{2}$, or
$\operatorname*{R}_{0}$\ better than for $\operatorname*{Q}_{2}$.

It is, however, illuminating to see what happens if we run the reduction of
Theorem \ref{lower}, starting from the assumption that $C$\ is a
\textit{zero-error} randomized or quantum algorithm for \textsc{WeakPar}%
$_{n,\varepsilon}$. \ Suppose furthermore that $C$ satisfies the same strong
property that our zero-error algorithms from Section \ref{ALG} satisfied:
namely, the property that $C$ halts after $T$\ queries in expectation on
\textit{every} input $X\in\left\{  0,1\right\}  ^{n}$. \ In that case, one can
skip the amplification step in Theorem \ref{lower}, to produce an algorithm
$C^{\prime}$\ with the following properties:

\begin{enumerate}
\item[(i)] $C$ halts after $T$\ queries in expectation on every input $X$, and

\item[(ii)] $C$ guesses \textsc{Par}$\left(  X\right)  $\ with probability
greater than $1/2$ on every input $X$.
\end{enumerate}

Now, one might think the above would imply $T\geq n/2$ (regardless of
$\varepsilon$), thereby contradicting our upper bounds from Section \ref{ALG}!
\ However, the apparent paradox is resolved once we realize that the lower
bound of Beals et al.\ \cite{bbcmw}---showing that $T\geq n/2$\ queries are
needed to guess \textsc{Par}$\left(  X\right)  $\ with probability greater
than $1/2$\ on every input $X$---says nothing about \textit{expected} query
complexity. \ And indeed, it is trivial to design an algorithm that guesses
\textsc{Par}$\left(  X\right)  $\ with $1/2+\varepsilon$\ probability on every
input $X$, using $2\varepsilon n$\ queries\ in expectation. \ That algorithm
just evaluates \textsc{Par}$\left(  X\right)  $\ (using $n$ queries) with
probability $2\varepsilon$, and otherwise guesses randomly, without examining
$X$ at all!

\section{Lower Bound via Sensitivity\label{SENSLB}}

Theorem \ref{lower} shows that our algorithms from Theorem \ref{upper}\ are
close to optimal when $\varepsilon$\ is reasonably large. \ Unfortunately,
though, Theorem \ref{lower}\ gives nothing when $\varepsilon=2^{-n}$.
\ Equivalently, it does not even rule out a randomized or quantum algorithm
making a \textit{constant} number of queries (!), that correctly decides
\textsc{Parity}\ on a subset of size $2^{n-1}+1$. \ We conjecture that
$n^{\Omega\left(  1\right)  }$\ randomized or quantum queries are needed for
the latter task, but we are unable to prove that conjecture---a state of
affairs that Section \ref{DQSEC} will help to explain. \ In this section, we
at least prove that $\Omega(\log n)$ randomized queries and $\Omega(\sqrt{\log
n})$ quantum queries are needed to solve \textsc{Weak Parity}\ for all
$\varepsilon>0$.

The key is a combinatorial quantity called $\Lambda\left(  n\right)  $, which
was introduced by Chung, F\"{u}redi, Graham, and Seymour \cite{chung}.
\ Abusing notation, we identify the set $\left\{  0,1\right\}  ^{n}$\ with the
Boolean hypercube graph (where two vertices are adjacent if and only if they
have Hamming distance $1$), and also identify any subset $G\subseteq\left\{
0,1\right\}  ^{n}$\ with the induced subgraph of $\left\{  0,1\right\}  ^{n}$
whose vertex set is $G$. Let $\Delta\left(  G\right)  $\ be the maximum degree
of any vertex in $G$. \ Then%
\begin{equation}
\Lambda\left(  n\right)  :=\min_{G\subseteq\left\{  0,1\right\}
^{n}~:~\left\vert G\right\vert =2^{n-1}+1}\Delta\left(  G\right)
\end{equation}
is the minimum of $\Delta\left(  G\right)  $\ over all induced subgraphs $G$
of size $2^{n-1}+1$.

The following proposition relates $\Lambda\left(  n\right)  $ to \textsc{Weak
Parity}.

\begin{proposition}
\label{chungprop}$\operatorname*{R}_{2}($\textsc{WeakPar}$_{n,\varepsilon
})=\Omega(\Lambda\left(  n\right)  )$ and $\operatorname*{Q}_{2}%
($\textsc{WeakPar}$_{n,\varepsilon})=\Omega(\sqrt{\Lambda\left(  n\right)  })$
for all $\varepsilon>0$.
\end{proposition}

\begin{proof}
Let $U$ be an algorithm that decides \textsc{Parity} (with bounded error
probability) on a subset $A\subseteq\left\{  0,1\right\}  ^{n}$. \ Then we
claim that $U$ must make $\Omega(\Delta\left(  A\right)  )$\ randomized or
$\Omega(\sqrt{\Delta\left(  A\right)  })$\ quantum queries, which is
$\Omega(\Lambda\left(  n\right)  )$\ or $\Omega(\sqrt{\Lambda\left(  n\right)
})$\ respectively if $\left\vert A\right\vert >2^{n-1}$. \ To see this, let
$X\in A$\ be a vertex with degree $\Delta\left(  A\right)  $. \ Then
\textsc{Parity}, when restricted to $X$ and its neighbors, already yields a
Grover search instance of size $\Delta\left(  A\right)  $. \ But searching a
list of $N$ elements is well-known to require $\Omega(N)$\ randomized or
$\Omega(\sqrt{N})$\ quantum queries \cite{bbbv}.
\end{proof}

To build intuition, it is easy to find an induced subgraph $G\subseteq\left\{
0,1\right\}  ^{n}$\ such that $\left\vert G\right\vert =2^{n-1}$\ but
$\Delta\left(  G\right)  =0$: consider the set of all points with odd Hamming
weight. \ But adding a single vertex to that $G$ increases its maximum degree
$\Delta\left(  G\right)  $\ all the way to $n$. \ More generally, Chung et
al.\ \cite{chung} were able to prove the following.\footnote{Chung et al.'s
result is very closely related to an earlier result of Simon \cite{simonsf},
which states that if $f:\left\{  0,1\right\}  ^{n}\rightarrow\left\{
0,1\right\}  $\ is a Boolean function depending on all $n$ of its inputs, then
$\operatorname*{s}\left(  f\right)  \geq\frac{1}{2}\log_{2}n-\frac{1}{2}%
\log_{2}\log_{2}n+\frac{1}{2}$, where $\operatorname*{s}\left(  f\right)
$\ is the sensitivity. \ However, neither Chung et al.'s result nor Simon's
seems derivable as an immediate corollary of the other.}

\begin{theorem}
[Chung et al.\ \cite{chung}]\label{chungthm}We have
\begin{equation}
\Lambda\left(  n\right)  \geq\frac{1}{2}\log_{2}n-\frac{1}{2}\log_{2}\log
_{2}n+\frac{1}{2}\,.
\end{equation}

\end{theorem}

Combining Theorem \ref{chungthm} with Proposition \ref{chungprop}\ tells us
immediately that%
\begin{align}
\operatorname*{R}\nolimits_{2}(\text{\textsc{WeakPar}}_{n,\varepsilon})  &
=\Omega(\log n),\\
\operatorname*{Q}\nolimits_{2}(\text{\textsc{WeakPar}}_{n,\varepsilon})  &
=\Omega(\sqrt{\log n})
\end{align}
for all $\varepsilon>0$.

Now, the best-known \textit{upper} bound on $\Lambda\left(  n\right)  $, also
proved by Chung et al.\ \cite{chung}, is $\sqrt{n}+1$, and it is conjectured
that this is essentially tight. \ By Proposition \ref{chungprop}, clearly a
proof of that conjecture would imply%
\begin{align}
\operatorname*{R}\nolimits_{2}(\text{\textsc{WeakPar}}_{n,\varepsilon})  &
=\Omega(\sqrt{n}),\\
\operatorname*{Q}\nolimits_{2}(\text{\textsc{WeakPar}}_{n,\varepsilon})  &
=\Omega\left(  n^{1/4}\right)
\end{align}
for all $\varepsilon>0$---and more generally, proving $\Lambda\left(
n\right)  \geq n^{\Omega\left(  1\right)  }$\ would imply that
$\operatorname*{R}($\textsc{WeakPar}$_{n,\varepsilon})$\ and
$\operatorname*{Q}($\textsc{WeakPar}$_{n,\varepsilon})$ are $n^{\Omega\left(
1\right)  }$.

Unfortunately, proving $\Lambda\left(  n\right)  \geq n^{\Omega\left(
1\right)  }$ will be challenging. \ To see why, recall the famous
\textit{Sensitivity Conjecture} (Conjecture \ref{sensconj}), which says that
$\operatorname*{s}\left(  f\right)  $\ is polynomially related to
$\operatorname*{bs}\left(  f\right)  $\ (or equivalently, to $\deg\left(
f\right)  $). \ In 1992, Gotsman and Linial \cite{gotsmanlinial} showed that
the Sensitivity Conjecture is equivalent to a statement about the maximum
degrees of induced subgraphs of $\left\{  0,1\right\}  ^{n}$:

\begin{theorem}
[Gotsman-Linial \cite{gotsmanlinial}]\label{glthm}Given any growth rate $h$,
we have $\operatorname*{s}\left(  f\right)  >h\left(  \deg\left(  f\right)
\right)  $ for all Boolean functions $f:\left\{  0,1\right\}  ^{n}%
\rightarrow\left\{  0,1\right\}  $, if and only if%
\begin{equation}
\max\left\{  \Delta\left(  G\right)  ,\Delta\left(  \left\{  0,1\right\}
^{n}\setminus G\right)  \right\}  \geq h\left(  n\right)
\end{equation}
for all subsets\ $G\subseteq\left\{  0,1\right\}  ^{n}$\ such that $\left\vert
G\right\vert \neq2^{n-1}$.
\end{theorem}

Notice that if $\left\vert G\right\vert \neq2^{n-1}$, then%
\begin{equation}
\max\left\{  \Delta\left(  G\right)  ,\Delta\left(  \left\{  0,1\right\}
^{n}\setminus G\right)  \right\}  \geq\Lambda\left(  n\right)  .
\end{equation}
To see this, choose whichever of $G$\ or $\left\{  0,1\right\}  ^{n}\setminus
G$\ is larger, and then discard all but $2^{n-1}+1$\ of its elements. \ Thus,
any lower bound on Chung et al.'s combinatorial quantity\ $\Lambda\left(
n\right)  $ implies the same lower bound on the function $h\left(  n\right)
$\ of Theorem \ref{glthm}. \ For example, if $\Lambda\left(  n\right)  \geq
n^{\Omega\left(  1\right)  }$, then $\operatorname*{s}\left(  f\right)
\geq\deg\left(  f\right)  ^{\Omega\left(  1\right)  }$.

But this means that \textit{any proof of }$\Lambda\left(  n\right)  \geq
n^{\Omega\left(  1\right)  }$\textit{\ would imply the Sensitivity
Conjecture!}\footnote{Interestingly, we do not know the reverse implication.}
\ Thus, the conjecture $\Lambda\left(  n\right)  \geq n^{\Omega\left(
1\right)  }$\ could be seen as a \textquotedblleft common combinatorial
core\textquotedblright\ of the \textsc{Weak Parity}\ and sensitivity versus
block sensitivity questions.

As a final note, Andy Drucker (personal communication) found a self-contained
proof for $\operatorname*{R}\nolimits_{2}($\textsc{WeakPar}$_{n,\varepsilon
})=\Omega(\log n)$, one not relying on $\Lambda\left(  n\right)  $, which we
include with Drucker's kind permission.

\begin{proposition}
[Drucker]\label{druckerprop}$\operatorname*{R}\nolimits_{2}($\textsc{WeakPar}%
$_{n,\varepsilon})=\Omega(\log n)$ for all $\varepsilon>0$.
\end{proposition}

\begin{proof}
Suppose by contradiction that (say) $\operatorname*{R}\nolimits_{2}%
($\textsc{WeakPar}$_{n,\varepsilon})\leq0.9\log_{2}n$, and let $U$ be a
randomized algorithm that achieves the bound. \ Then we can think of $U$ as
just a probability distribution $\mathcal{D}$\ over decision trees $T$.
\ Given a decision tree $T$, let $S_{T}\subseteq\left[  n\right]  $\ be the
set of all indices $i$ such that the variable $x_{i}$\ appears anywhere in
$T$. \ Then by assumption, each $T$ in the support of $\mathcal{D}$ has depth
at most $0.9\log_{2}n$, and therefore satisfies $\left\vert S_{T}\right\vert
\leq2^{0.9\log_{2}n}=n^{0.9}$. \ By averaging, it follows that there exists an
$i\in\left[  n\right]  $\ such that%
\begin{equation}
\Pr_{T\sim\mathcal{D}}\left[  i\in S_{T}\right]  \leq\frac{1}{n^{0.1}}.
\end{equation}
But this means that, for every $X\in\left\{  0,1\right\}  ^{n}$, we must have%
\[
\left\vert \Pr\left[  U\text{ accepts }X\right]  -\Pr\left[  U\text{ accepts
}X^{\left\{  i\right\}  }\right]  \right\vert \leq\frac{1}{n^{0.1}},
\]
where $X^{\left\{  i\right\}  }$\ denotes $X$ with the $i^{th}$\ bit flipped
(as in Section \ref{SBS}). \ Hence either $\Pr\left[  U\text{ accepts
}X\right]  $\ fails to approximate \textsc{Par}$\left(  X\right)  $, or else
$\Pr\left[  U\text{ accepts }X^{\left\{  i\right\}  }\right]  $\ fails to
approximate \textsc{Par}$(X^{\left\{  i\right\}  })$. \ But this means that
$U$ weakly computes \textsc{Parity}\ on at most $2^{n-1}$\ inputs.
\end{proof}

Interestingly, unlike with our argument based on $\Lambda\left(  n\right)  $,
we do not know how to generalize Drucker's argument to prove any lower bound
on \textit{quantum} query complexity, nor do we know (even conjecturally) how
the argument might be pushed beyond $\Omega(\log n)$.

\section{\label{DQSEC}Connection to $\deg\left(  f\right)  $\ vs.
$\operatorname*{Q}\left(  f\right)  $}

In the last section, we identified a known combinatorial conjecture
($\Lambda\left(  n\right)  \geq n^{\Omega\left(  1\right)  }$) that would
imply that the randomized and quantum query complexities of \textsc{Weak
Parity}\ are $n^{\Omega\left(  1\right)  }$\ for all $\varepsilon>0$.
\ However, since $\Lambda\left(  n\right)  \geq n^{\Omega\left(  1\right)  }$
would also imply the Sensitivity Conjecture, it will clearly be difficult to prove.

So could there be a \textit{different} way to prove tight lower bounds for
$\operatorname*{R}\nolimits_{2}($\textsc{WeakPar}$_{n,\varepsilon})$\ and
$\operatorname*{Q}\nolimits_{2}($\textsc{WeakPar}$_{n,\varepsilon})$---a way
that wouldn't require us to address any longstanding open problems about
Boolean functions? \ Alas, in this section we largely close off that
possibility. \ In particular, suppose we could prove a strong lower bound on
$\operatorname*{R}\nolimits_{2}($\textsc{WeakPar}$_{n,\varepsilon})$. \ We
will show that this would imply a better polynomial relationship between
$\deg\left(  f\right)  $\ and $\operatorname*{R}\nolimits_{2}\left(  f\right)
$ for \textit{all} total Boolean functions $f$ than is currently known.
\ Similar statements hold for $\operatorname*{R}\nolimits_{0}$,
$\operatorname*{Q}\nolimits_{2}$, and $\operatorname*{Q}\nolimits_{0}$.

\begin{theorem}
\label{ambthm}Given a constant $c$, suppose there exists a sequence of
functions $\left\{  f_{n}\right\}  _{n\geq1}$\ such that $\deg\left(
f_{n}\right)  =n$\ and $\operatorname*{R}\nolimits_{2}\left(  f_{n}\right)
=O\left(  n^{c}\right)  $. \ Then%
\begin{equation}
\operatorname*{R}\nolimits_{2}(\text{\textsc{WeakPar}}_{n,\varepsilon
})=O\left(  \frac{n}{\log^{1-c}1/\varepsilon}\right)  .
\end{equation}
The same holds if we replace $\operatorname*{R}\nolimits_{2}$\ by
$\operatorname*{R}\nolimits_{0}$, $\operatorname*{Q}\nolimits_{2}$, or
$\operatorname*{Q}\nolimits_{0}$ in both instances.
\end{theorem}

\begin{proof}
We first show that $\operatorname*{R}\nolimits_{2}($\textsc{WeakPar}%
$_{n,2^{-n}})=O\left(  n^{c}\right)  $ in the special case $\varepsilon
=2^{-n}$; then we generalize to larger $\varepsilon$.

Observe that we can assume without loss of generality that each $f_{n}$\ has
exactly\ $n$ inputs. \ For otherwise, let $p$ be the unique multilinear
polynomial representing $f_{n}$; then choose a monomial $m$\ of $p$ with
degree $n$, and arbitrarily fix all bits that do not appear in $m$. \ This
yields a subfunction $f_{n}^{\prime}$\ with $n$ inputs, $\deg\left(
f_{n}^{\prime}\right)  =\deg\left(  f_{n}\right)  =n$, and $\operatorname*{R}%
\nolimits_{2}\left(  f_{n}^{\prime}\right)  \leq\operatorname*{R}%
\nolimits_{2}\left(  f_{n}\right)  $.

Now by Proposition \ref{degreeprop}, the statement $\deg\left(  f_{n}\right)
=n$\ is equivalent to the combinatorial statement%
\begin{equation}
\left\vert \left\{  X:f_{n}\left(  X\right)  =1\text{ and \textsc{Par}}\left(
X\right)  =0\right\}  \right\vert \neq\left\vert \left\{  X:f_{n}\left(
X\right)  =1\text{ and \textsc{Par}}\left(  X\right)  =1\right\}  \right\vert
.
\end{equation}
This means that $f_{n}\left(  X\right)  $\ either agrees or disagrees with
\textsc{Par}$\left(  X\right)  $\ on at least $2^{n-1}+1$\ inputs $X$. \ By
replacing $f_{n}$\ by $1-f_{n}$, we can assume without loss of generality that
the first case holds. \ Then if we run the algorithm for $f_{n}$, it will make
$O\left(  n^{c}\right)  $\ queries and correctly decide \textsc{Parity}\ on at
least $2^{n-1}+1$\ inputs, which was the desired result.

To generalize to arbitrary $\varepsilon$, we simply need to appeal to Lemma
\ref{epsrelation}, which tells us that if%
\begin{equation}
\operatorname*{R}\nolimits_{2}(\text{\textsc{WeakPar}}_{n,2^{-n}})\leq
T\left(  n\right)  =O\left(  n^{c}\right)  ,
\end{equation}
then%
\begin{equation}
\operatorname*{R}\nolimits_{2}(\text{\textsc{WeakPar}}_{N,\varepsilon}%
)\leq\frac{N\cdot T\left(  \log_{2}1/\varepsilon\right)  }{\log_{2}%
1/\varepsilon}=O\left(  \frac{N}{\log^{1-c}1/\varepsilon}\right)  .
\end{equation}
Finally, since we never used that the algorithm was classical or
bounded-error, everything in the proof still works if we replace
$\operatorname*{R}\nolimits_{2}$\ by $\operatorname*{R}\nolimits_{0}$,
$\operatorname*{Q}\nolimits_{2}$, or $\operatorname*{Q}\nolimits_{0}$\ throughout.
\end{proof}

For clarity, let us state Theorem \ref{ambthm}\ in contrapositive form.

\begin{corollary}
\label{ambcor}Suppose $\operatorname*{R}\nolimits_{2}($\textsc{WeakPar}%
$_{n,\varepsilon})=\Omega\left(  n/\log^{1-c}\left(  1/\varepsilon\right)
\right)  $. \ Then for every Boolean function $f$, we have $\operatorname*{R}%
\nolimits_{2}\left(  f\right)  =\Omega\left(  \deg\left(  f\right)
^{c}\right)  $\ (and similarly for $\operatorname*{R}\nolimits_{0}$,
$\operatorname*{Q}\nolimits_{2}$, and $\operatorname*{Q}\nolimits_{0}$).
\end{corollary}

Plugging our $\Omega\left(  n/\log\left(  1/\varepsilon\right)  \right)
$\ lower bound on $\operatorname*{R}\nolimits_{2}($\textsc{WeakPar}%
$_{n,\varepsilon})$\ (i.e., Theorem \ref{lower}) into Corollary \ref{ambcor},
we get only the trivial lower bound $\operatorname*{R}\nolimits_{2}\left(
f\right)  =\Omega\left(  1\right)  $ for non-constant $f$. \ On the other
hand, suppose we could prove that%
\begin{equation}
\operatorname*{R}\nolimits_{2}(\text{\textsc{WeakPar}}_{n,\varepsilon}%
)=\Omega\left(  \frac{n}{\log^{2/3}1/\varepsilon}\right)  .
\end{equation}
Then Corollary \ref{ambcor} would reproduce the result of Nisan \cite{nisan}
that $\operatorname*{R}\nolimits_{2}\left(  f\right)  =\Omega(\deg\left(
f\right)  ^{1/3})$ for all Boolean functions $f$. \ Likewise, if we could
prove that%
\begin{equation}
\operatorname*{Q}\nolimits_{2}(\text{\textsc{WeakPar}}_{n,\varepsilon}%
)=\Omega\left(  \frac{n}{\log^{5/6}1/\varepsilon}\right)  ,
\end{equation}
then Corollary \ref{ambcor}\ would reproduce the result of Beals et
al.\ \cite{bbcmw}\ that $\operatorname*{Q}\nolimits_{2}\left(  f\right)
=\Omega(\deg\left(  f\right)  ^{1/6})$\ for all $f$. \ Any \textit{better}
lower bounds than those on $\operatorname*{R}\nolimits_{2}($\textsc{WeakPar}%
$_{n,\varepsilon})$\ or $\operatorname*{Q}\nolimits_{2}($\textsc{WeakPar}%
$_{n,\varepsilon})$\ would imply better general lower bounds on
$\operatorname*{R}\nolimits_{2}\left(  f\right)  $ or $\operatorname*{Q}%
_{2}\left(  f\right)  $\ than are currently known. \ So for example, suppose
we could prove that%
\begin{equation}
\operatorname*{Q}\nolimits_{2}(\text{\textsc{WeakPar}}_{n,\varepsilon}%
)=\Omega\left(  \frac{n}{\sqrt{\log1/\varepsilon}}\right)  ;
\end{equation}
i.e., that the quantum algorithm of Theorem \ref{upper}\ was optimal. \ Then
we\ would prove the longstanding conjecture that $\operatorname*{Q}_{2}\left(
f\right)  =\Omega(\sqrt{\deg\left(  f\right)  })$\ for all Boolean functions
$f$ (the bound being saturated when $f=\operatorname*{OR}$).

One might wonder: can we also go in the other direction, and use the known
polynomial relationships between $\deg\left(  f\right)  $ and query complexity
measures\ to prove better lower bounds for \textsc{Weak Parity}? \ At present,
we cannot quite do that, but we can do something close. \ Recall from Section
\ref{INTRO} that, in defining \textsc{Weak Parity}, we did not impose any
requirement that our algorithm's acceptance probability $p\left(  X\right)
$\ approximate a total Boolean function. \ However, suppose we \textit{do}
impose that requirement. \ Then we can easily show the following:

\begin{proposition}
\label{totalprop}Fix any $\varepsilon>0$. \ Suppose an algorithm's acceptance
probability must satisfy $p\left(  X\right)  \in\left[  0,1/3\right]
\cup\left[  2/3,1\right]  $\ for all $X\in\left\{  0,1\right\}  ^{n}$. \ Then
any randomized algorithm for \textsc{WeakPar}$_{n,\varepsilon}$\ makes
$\Omega\left(  n^{1/3}\right)  $\ queries, and any quantum algorithm\ makes
$\Omega\left(  n^{1/6}\right)  $ queries.

Suppose further that the acceptance probability must satisfy $p\left(
X\right)  \in\left\{  0,1\right\}  $\ for all $X$. \ Then any randomized
algorithm for \textsc{WeakPar}$_{n,\varepsilon}$\ makes $\Omega\left(
n^{1/2}\right)  $\ queries in expectation, and any quantum algorithm\ makes
$\Omega\left(  n^{1/3}\right)  $ queries.
\end{proposition}

\begin{proof}
Let $f\left(  X\right)  =\left\lfloor p\left(  X\right)  \right\rceil $ be the
total Boolean function approximated by $p\left(  X\right)  $. \ Then since the
algorithm solves \textsc{Weak Parity},%
\begin{equation}
\left\vert \left\{  X:f\left(  X\right)  =1\text{ and \textsc{Par}}\left(
X\right)  =1\right\}  \right\vert >\left\vert \left\{  X:f\left(  X\right)
=1\text{ and \textsc{Par}}\left(  X\right)  =0\right\}  \right\vert .
\end{equation}
So by Proposition \ref{degreeprop},\ we must have $\deg\left(  f\right)
=\operatorname*{D}\left(  f\right)  =n$. \ By Theorems \ref{drthm},
\ref{bbcmwthm}, and \ref{midrijanisthm}, this means that%
\begin{align}
\operatorname*{R}\nolimits_{2}\left(  f\right)   &  =\Omega(\operatorname*{D}%
\left(  f\right)  ^{1/3})=\Omega(n^{1/3}),\\
\operatorname*{Q}\nolimits_{2}\left(  f\right)   &  =\Omega(\operatorname*{D}%
\left(  f\right)  ^{1/6})=\Omega(n^{1/6}),\\
\operatorname*{R}\nolimits_{0}\left(  f\right)   &  =\Omega(\operatorname*{D}%
\left(  f\right)  ^{1/2})=\Omega(n^{1/2}),\\
\operatorname*{Q}\nolimits_{0}\left(  f\right)   &  =\Omega(\operatorname*{D}%
\left(  f\right)  ^{1/3})=\Omega(n^{1/3}).
\end{align}

\end{proof}

Of course, any improvement to the known polynomial relationships between
$\operatorname*{D}\left(  f\right)  $\ and $\operatorname*{R}\nolimits_{2}%
\left(  f\right)  $, $\operatorname*{Q}\nolimits_{2}\left(  f\right)  $,
etc.\ for total Boolean $f$ would automatically yield a corresponding
improvement to Proposition \ref{totalprop}.

\section{Weak Algorithms for Other Functions\label{OTHER}}

In this section, we begin the investigation of weak algorithms for Boolean
functions other than \textsc{Parity}. \ Our main result is the following:

\begin{theorem}
\label{weakf}Let $f:\left\{  0,1\right\}  ^{n}\rightarrow\left\{  0,1\right\}
$ be any Boolean function. \ Then we can guess $f\left(  X\right)  $\ on
$2^{n-1}+1$ inputs $X$ using a bounded-error quantum algorithm that makes
$O(\sqrt{n})$ queries, a zero-error randomized algorithm that makes $O\left(
n^{0.754}\right)  $\ queries, or a zero-error quantum algorithm that makes
$O(\sqrt{n}\log^{2}n)$\ queries.
\end{theorem}

\begin{proof}
Assume without loss of generality that $n$\ has the form $2^{d}$\ for $d\geq
2$\ (this will not affect the asymptotics). \ There are two cases. \ First,
suppose $f$ is \textit{unbalanced}: that is,%
\begin{equation}
\left\vert \left\{  X:f\left(  X\right)  =1\right\}  \right\vert \neq2^{n-1}.
\end{equation}
Then we can trivially agree with $f$ on at least $2^{n-1}+1$\ inputs $X$, by
either always outputting $0$ or always outputting $1$.

Second, suppose $f$\ is balanced. \ Note that the $\operatorname*{OR}%
$\ function outputs $1$ on an odd number of inputs $X$. \ It follows that
$\left\vert \left\{  X:f\left(  X\right)  =\operatorname*{OR}\left(  X\right)
\right\}  \right\vert $\ must be odd as well, and cannot equal $2^{n-1}$. \ So
either $\operatorname*{OR}\left(  X\right)  $\ or $1-\operatorname*{OR}\left(
X\right)  $\ must agree with $f\left(  X\right)  $\ on at least $2^{n-1}%
+1$\ inputs $X$. \ Thus, Grover's algorithm gives us the desired bounded-error
quantum algorithm making $O(\sqrt{n})$ queries.

For the other algorithms, recall Proposition \ref{tvspar}, which tells us that
the AND/OR tree $\operatorname*{T}_{d}$\ also outputs $1$ on an odd number of
inputs $X$. \ So by the same reasoning as above, either $\operatorname*{T}%
_{d}\left(  X\right)  $\ or $1-\operatorname*{T}_{d}\left(  X\right)  $\ must
agree with $f\left(  X\right)  $\ on at least $2^{n-1}+1$\ inputs $X$.
\ Hence, we can use Theorem \ref{swthm}\ to get the desired zero-error
randomized algorithm making $O\left(  n^{0.754}\right)  $\ queries, and use
Theorem \ref{acltthm}\ to get a zero-error quantum algorithm making
$O(\sqrt{n}\log^{2}n)$\ queries.
\end{proof}

Interestingly, unlike for \textsc{Parity}, for arbitrary $f$ it is unclear
whether we can get any nontrivial algorithms when $\varepsilon$ is larger than
$2^{-n}$. \ Our proof of Lemma \ref{epsrelation} relied essentially on
\textsc{Parity}'s downward self-reducibility, so it does not generalize to
other functions.

Note also that we cannot hope to prove any general \textit{lower} bound on the
weak query complexity of $f$, even assuming that $f$ is balanced and that its
quantum query complexity is $\Omega\left(  n\right)  $. \ As a counterexample,
let $\operatorname*{H}\left(  X\right)  =1$\ if $X$ has Hamming weight at
least $2n/3$\ and $\operatorname*{H}\left(  X\right)  =0$\ otherwise; then
consider%
\begin{equation}
f\left(  x_{1},\ldots,x_{n}\right)  :=x_{1}\oplus\operatorname*{H}\left(
x_{2},\ldots,x_{n}\right)  .
\end{equation}

\section{Open Problems\label{CONC}}

The obvious problem is to close the gaps between our upper and lower bounds on
the query complexity of \textsc{Weak Parity}. \ We have seen that this problem
is intimately related to longstanding open problems in the study of Boolean
functions, including polynomial degree versus query complexity, the
Sensitivity Conjecture, and lower-bounding Chung et al.'s \cite{chung}%
\ combinatorial quantity $\Lambda\left(  n\right)  $.\ \ Perhaps the
surprising relationships among these problems could motivate renewed attacks.

In the meantime, can we reprove our $\Omega\left(  n/\log\left(
1/\varepsilon\right)  \right)  $\ lower bound for \textsc{Weak Parity}\ (or
better yet, improve it)\ \textit{without} exploiting \textsc{Parity}'s random
self-reducibility? \ How far can we get by using (say) the polynomial or
adversary methods directly? \ It would also be great if we could say something
about weak algorithms for functions other than \textsc{Parity}, beyond what we
said in Section \ref{OTHER}: for example, what happens if $\varepsilon>2^{-n}$?

Let us end with three more specific questions:

\begin{enumerate}
\item[(1)] Do we ever get faster algorithms for \textsc{Weak Parity}, if we
drop the constraint that the algorithm's acceptance probability approximates a
total Boolean function $f$?\footnote{If the answer is no, then of course
Proposition \ref{totalprop}\ already gives us good lower bounds for
\textsc{Weak Parity}\ when $\varepsilon=2^{-n}$.}

\item[(2)] Can we \textquotedblleft interpolate\textquotedblright\ between our
two different ways of proving lower bounds for \textsc{Weak Parity}, to get
better lower bounds than $\Omega(\log n)$\ or $\Omega(\sqrt{\log n})$\ when
$\varepsilon$\ is small but still larger than $2^{-n}$?

\item[(3)] Can we show that an $n^{\Omega\left(  1\right)  }$\ lower bound for
\textsc{Weak Parity}\ is directly implied by the Sensitivity Conjecture,
rather than the related conjecture that $\Lambda\left(  n\right)  \geq
n^{\Omega\left(  1\right)  }$?
\end{enumerate}

\section{Acknowledgments}

We thank Ronald de Wolf, both for helpful discussions and for his comments on
a draft, and Andy Drucker for allowing us to include Proposition
\ref{druckerprop}.

\bibliographystyle{plain}
\bibliography{thesis}

\begin{thebibliography}{10}

\bibitem{aar:bf}
S.~Aaronson.
\newblock Algorithms for {B}oolean function query properties.
\newblock {\em SIAM J. Comput.}, 32(5):1140--1157, 2003.

\bibitem{aar:cer}
S.~Aaronson.
\newblock Quantum certificate complexity.
\newblock In {\em Proc. IEEE Conference on Computational Complexity}, pages
  171--178, 2003.
\newblock ECCC TR03-005, quant-ph/0210020.

\bibitem{aclt}
A.~Ambainis, A.~Childs, F.~Le Gall, and S.~Tani.
\newblock The quantum query complexity of certification.
\newblock {\em Quantum Information and Computation}, 10(3-4), 2010.
\newblock arXiv:0903.1291.

\bibitem{acrsz}
A.~Ambainis, A.~M. Childs, B.~W. Reichardt, R.~\v{S}palek, and S.~Zhang.
\newblock Any {AND-OR} formula of size {$N$} can be evaluated in time
  {$N^{1/2+o(1)}$} on a quantum computer.
\newblock In {\em Proc. IEEE FOCS}, 2007.
\newblock quant-ph/0703015 and arXiv:0704.3628.

\bibitem{bbcmw}
R.~Beals, H.~Buhrman, R.~Cleve, M.~Mosca, and R.~de Wolf.
\newblock Quantum lower bounds by polynomials.
\newblock {\em J. ACM}, 48(4):778--797, 2001.
\newblock Earlier version in IEEE FOCS 1998, pp. 352-361. quant-ph/9802049.

\bibitem{bbbv}
C.~Bennett, E.~Bernstein, G.~Brassard, and U.~Vazirani.
\newblock Strengths and weaknesses of quantum computing.
\newblock {\em SIAM J. Comput.}, 26(5):1510--1523, 1997.
\newblock quant-ph/9701001.

\bibitem{bcwz}
H.~Buhrman, R.~Cleve, R.~de Wolf, and Ch. Zalka.
\newblock Bounds for small-error and zero-error quantum algorithms.
\newblock In {\em Proc. IEEE FOCS}, pages 358--368, 1999.
\newblock cs.CC/9904019.

\bibitem{bw}
H.~Buhrman and R.~de Wolf.
\newblock Complexity measures and decision tree complexity: a survey.
\newblock {\em Theoretical Comput. Sci.}, 288:21--43, 2002.

\bibitem{chung}
F.~R.~K. Chung, Z.~F\"{u}redi, R.~L. Graham, and P.~Seymour.
\newblock On induced subgraphs of the cube.
\newblock {\em J. Comb. Theory Ser. A}, 49(1):180--187, 1988.

\bibitem{dj}
D.~Deutsch and R.~Jozsa.
\newblock Rapid solution of problems by quantum computation.
\newblock {\em Proc. Roy. Soc. London}, A439:553--558, 1992.

\bibitem{fgg:nand}
E.~Farhi, J.~Goldstone, and S.~Gutmann.
\newblock A quantum algorithm for the {H}amiltonian {NAND} tree.
\newblock {\em Theory of Computing}, 4(1):169--190, 2008.
\newblock quant-ph/0702144.

\bibitem{farhi:parity}
E.~Farhi, J.~Goldstone, S.~Gutmann, and M.~Sipser.
\newblock A limit on the speed of quantum computation in determining parity.
\newblock {\em Phys. Rev. Lett.}, 81:5442--5444, 1998.
\newblock quant-ph/9802045.

\bibitem{gotsmanlinial}
C.~Gotsman and N.~Linial.
\newblock The equivalence of two problems on the cube.
\newblock {\em J. Comb. Theory Ser. A}, 61(1):142--146, 1992.

\bibitem{hatami}
P.~Hatami, R.~Kulkarni, and D.~Pankratov.
\newblock Variations on the sensitivity conjecture.
\newblock {\em Theory of Computing Library Graduate Surveys}, 4, 2011.

\bibitem{midrijanis:drq}
G.~Midrijanis.
\newblock On randomized and quantum query complexities.
\newblock quant-ph/0501142, 2005.

\bibitem{nisan}
N.~Nisan.
\newblock {CREW PRAM}s and decision trees.
\newblock {\em SIAM J. Comput.}, 20(6):999--1007, 1991.

\bibitem{ns}
N.~Nisan and M.~Szegedy.
\newblock On the degree of {B}oolean functions as real polynomials.
\newblock {\em Computational Complexity}, 4(4):301--313, 1994.

\bibitem{rubinstein}
D.~Rubinstein.
\newblock Sensitivity vs.\ block sensitivity of {B}oolean functions.
\newblock {\em Combinatorica}, 15(2):297--299, 1995.

\bibitem{sakswigderson}
M.~Saks and A.~Wigderson.
\newblock Probabilistic {B}oolean decision trees and the complexity of
  evaluating game trees.
\newblock In {\em Proc. IEEE FOCS}, pages 29--38, 1986.

\bibitem{santha}
M.~Santha.
\newblock On the {M}onte-{C}arlo decision tree complexity of read-once
  formulae.
\newblock {\em Random Structures and Algorithms}, 6(1):75--87, 1995.

\bibitem{simonsf}
H.-U. Simon.
\newblock A tight {$\Omega(\log \log n)$} bound on the time for parallel
  {RAM}'s to compute non-degenerate {B}oolean functions.
\newblock In {\em Foundations of Computing Theory}, volume 158, pages 439--444.
  Springer-Verlag, 1983.

\bibitem{tal}
A.~Tal.
\newblock Properties and applications of {B}oolean function composition.
\newblock In {\em Proc. Innovations in Theoretical Computer Science (ITCS)},
  pages 441--454, 2013.
\newblock ECCC TR12-163.

\end{thebibliography}

\end{document}